\documentclass[12pt,a4paper]{article}
\usepackage[margin=25mm]{geometry}
\usepackage{graphicx}
\usepackage[T1]{fontenc}
\usepackage{lmodern,amsmath,amssymb,amsthm,mathtools,booktabs,array,microtype}
\usepackage[numbers,sort&compress]{natbib}
\usepackage[hidelinks]{hyperref}
\newtheorem{theorem}{Theorem}
\newtheorem{corollary}{Corollary}
\newtheorem{proposition}{Proposition}
\theoremstyle{remark}
\newtheorem{example}{Example}
\newcommand{\R}{\mathbb R}
\newcommand{\Q}{\mathcal Q}
\newcommand{\ev}[1]{\langle #1\rangle}
\newcommand{\lr}{\le_{\mathrm{lr}}}
\newcommand{\doi}[1]{\href{https://doi.org/#1}{doi:\nolinkurl{#1}}}
\DeclareMathOperator{\tr}{tr}
\hypersetup{pdftitle={Beyond energy bounds: comparing Schrodinger eigenstates},
 pdfauthor={Xiao-Ke Wang, Han Jia, Shengjun Wu, Zeng-Bing Chen}}
\begin{document}
\title{Beyond energy bounds: comparing Schr\"odinger eigenstates}
\author{Xiao-Ke Wang, Han Jia, Shengjun Wu\thanks{Correspondence: \href{mailto:sjwu@nju.edu.cn}{sjwu@nju.edu.cn}.}
\ and Zeng-Bing Chen\thanks{Correspondence: \href{mailto:zbchen@nju.edu.cn}{zbchen@nju.edu.cn}.}\\[0.6em]
\small National Laboratory of Solid State Microstructures and School of Physics,\\
\small Collaborative Innovation Center of Advanced Microstructures,\\
\small Nanjing University, Nanjing 210093, China}
\date{}
\maketitle
\begin{abstract}
We derive eigenstate comparisons from ordered Schr\"odinger operators, extending energy bounds to overlaps, spatial distributions and nodes.
For three form-ordered operators with discrete comparison ground states, suppose the upper ground-state energy lies in the lower operator's ground-state gap.
Three comparison spectral values then bound all three squared ground-state overlaps and guarantee a simple discrete intermediate ground state.
Related gap estimates control low-energy subspaces and individual excited-state overlaps under additional gap conditions.
In one dimension and radially, a single crossing of the potential difference relative to the energy shift orders ground-state likelihood ratios and monotone observables; a nondecreasing difference suffices.
The results cover asymmetric wells and shortened right Dirichlet walls, with coordinatewise comparisons for submodular reference potentials and coordinatewise nondecreasing perturbations.
A two-sided Sturm argument orders corresponding internal nodes for nondecreasing potential differences, even between orthogonal one-dimensional excited states.
For Yukawa ground states, Hulth\'en comparisons bound overlaps, radial moments and the density at the nucleus; the latter two-sided bounds hold for screening $0<\delta<2/\sqrt6$.
Optimizing scaled Hulth\'en majorants gives a lower bound of approximately $1.13611$ for $n^2\delta_c(n,0)$, where $\delta_c(n,0)$ is the critical screening for the $n$th $s$-wave bound state.
This bound is optimal within that two-parameter comparison family.
\end{abstract}


\section{Introduction}
\label{sec:intro}
Two solvable potentials can bound the energies of a less tractable quantum system.
For $H_i=T+V_i$ with a common kinetic energy, $V_1\le V_2\le V_3$ orders the min--max values counted with multiplicity~\cite[section~11]{Simon2019Kato}.
We ask when a comparison ground state approximates the target, how spatial observables change, and where excited-state nodes move.

Energy-excess estimates control ground-state overlap~\cite{Eckart1930} and subspace leakage~\cite{VuLei2013}; perturbation theory bounds eigenspace rotation through a perturbation size and a spectral gap~\cite{DavisKahan1970,Seelmann2019}.
Spatial comparisons instead use structure in the potentials.
Wavefunction ratios and stochastic order give ground-state concentration for suitable reflection-symmetric perturbations~\cite{AHS1981,HerbstPitt1991}; Herbst and Pitt also develop a broader full-space diffusion framework.
For excited states, Sturm comparison controls zeros of one-dimensional differential equations~\cite{Simon2005}, and potential derivatives of nodes are known on finite intervals~\cite{GuoZhang2022}.

Our spectral estimates replace unknown target-state expectations by the comparison spectra.
The spatial results include asymmetric intervals, a shortened right wall and coordinatewise comparisons under submodularity; the even-potential mechanism is known~\cite{AHS1981,HerbstPitt1991}.
For excited states, we give a direct Sturm proof of node order and an example in which ordered nodes coexist with orthogonal states.

For Yukawa ground states, we build on Yukawa--Hulth\'en energy comparisons~\cite{Hall1985,Hall1992YH}, choosing one pair for energy and overlap bounds and another for radial moments and the density at the nucleus.
Optimizing the majorant's strength and screening gives the largest guaranteed binding range within that two-parameter family.

Table~\ref{tab:guide} summarizes the information supplied by the comparisons.
\begin{table}[t]
\caption{Inputs and outputs of the comparisons. The indicated sections give the full hypotheses.}
\label{tab:guide}
\centering
\begingroup
\fontsize{10}{12}\selectfont
\begin{tabular}{@{}>{\raggedright\arraybackslash}p{0.23\linewidth}>{\raggedright\arraybackslash}p{0.43\linewidth}>{\raggedright\arraybackslash}p{0.27\linewidth}@{}}
\toprule
Information sought & Additional input & Output and location\\
\midrule
Ground-state overlaps & Form order, two discrete comparison ground states and a positive margin from three spectral values & Three squared-overlap bounds; section~\ref{sec:overlaps}\\
Low-energy subspaces & Discrete low levels, exterior spectral sums and reference gaps & Total subspace leakage; section~\ref{sec:overlaps}\\
1D and radial observables & Nondecreasing difference, or single crossing relative to the ground-state energy shift & Likelihood-ratio order; section~\ref{sec:shape}\\
Coordinate observables & Submodular reference potential and coordinatewise nondecreasing perturbation & Coordinatewise comparison; section~\ref{sec:shape}\\
Internal nodes & Common interval, same bound-state index, nondecreasing difference & Corresponding node order; section~\ref{sec:nodes}\\
Yukawa application & Task-specific Hulth\'en comparison potentials & Explicit state bounds and optimized binding range; section~\ref{sec:yukawa}\\
\bottomrule
\end{tabular}
\endgroup
\end{table}

\section{Overlap bounds from comparison spectra}
\label{sec:overlaps}
\subsection{Setting and spectral notation}
Let $T=-\hbar^2\nabla^2/(2m)$.
On an open set $\Omega_i\subseteq\R^d$, impose Dirichlet conditions and define
\begin{equation}
 h_i[\phi]=\frac{\hbar^2}{2m}\int_{\Omega_i}|\nabla\phi|^2+
 \int_{\Omega_i}V_i|\phi|^2,
 \qquad
 \Q(h_i)=\{\phi\in H^1_0(\Omega_i):\int V_{i,+}|\phi|^2<\infty\}.
 \label{eq:form}
\end{equation}
Here $V_{i,\pm}=\max\{\pm V_i,0\}$. We assume $V_{i,+}\in L^1_{\mathrm{loc}}$ and that $V_{i,-}$ is relatively form bounded with respect to $T$ with relative bound less than one.
We call such potentials admissible. These assumptions define closed, lower-bounded forms and self-adjoint operators~\cite{ReedSimonII}.
Functions on smaller domains are extended by zero.
The order $h_1\le h_2\le h_3$ means
\begin{equation}
 \begin{aligned}
 &\Omega_3\subseteq\Omega_2\subseteq\Omega_1,\qquad
 \Q(h_3)\subseteq\Q(h_2)\subseteq\Q(h_1),\\
 &h_1[\phi]\le h_2[\phi]\le h_3[\phi]\qquad(\phi\in\Q(h_3)),
 \end{aligned}
 \label{eq:order}
\end{equation}
with the analogous pairwise inequalities on $\Q(h_2)$ and $\Q(h_3)$.
Pointwise ordering of the potentials on the nested domains implies this form order.

Write $\mu_n^{(i)}$ for the $n$th min--max value, $\Sigma_i=\inf\sigma_{\rm ess}(H_i)$ for the bottom of the essential spectrum, and $E_n^{(i)}=\mu_n^{(i)}$ when $\mu_n^{(i)}<\Sigma_i$.
Thus $E_n^{(i)}$ denotes an actual discrete eigenvalue, whereas $\mu_n^{(i)}$ can equal a continuum threshold.
For example, a three-dimensional well with exactly one negative eigenvalue and no zero-energy eigenstate has $E_1<0$ and $\mu_2=\Sigma=0$; the latter is a threshold rather than an eigenvalue.
The min--max principle gives $\mu_n^{(1)}\le\mu_n^{(2)}\le\mu_n^{(3)}$.
All eigenvalues are counted with multiplicity.
We denote normalized eigenstates by $\psi_n^{(i)}$, ground states by $\psi^{(i)}$, and their squared overlaps by
$F_{ij}=|\langle\psi^{(i)},\psi^{(j)}\rangle|^2$.

\subsection{Ground-state overlaps from three spectral values}
The first bound needs no spatial monotonicity: it uses the comparison ground-state energies and the lower operator's spectral gap.
Set $g_1=\mu_2^{(1)}-E_1^{(1)}>0$ and
\begin{equation}
 \kappa=\frac{\mu_2^{(1)}-E_1^{(3)}}{g_1}.
 \label{eq:kappa}
\end{equation}
\begin{theorem}\label{thm:overlap}
Suppose \eqref{eq:order} holds and $H_1,H_3$ have discrete ground states.
If $\kappa>0$, then $H_2$ has a simple discrete ground state, with
\begin{equation}
 \mu_2^{(2)}-E_1^{(2)}\ge\mu_2^{(1)}-E_1^{(3)},\qquad
 F_{12},F_{13},F_{23}\ge\kappa.
 \label{eq:alloverlaps}
\end{equation}
If all three discrete ground states exist without $\kappa>0$, the bounds $F_{12},F_{13}\ge\kappa$ still hold.
More precisely, whenever the intermediate discrete ground state exists,
\begin{equation}
 F_{12}\ge\frac{\mu_2^{(1)}-h_1[\psi^{(2)}]}{g_1}
 =\frac{\mu_2^{(1)}-E_1^{(2)}+(h_2-h_1)[\psi^{(2)}]}{g_1}\ge\kappa.
 \label{eq:refinedoverlap}
\end{equation}
The statements hold within any common reducing symmetry sector.
\end{theorem}

\begin{proof}
First, if $\kappa>0$, min--max gives $\mu_1^{(2)}\le E_1^{(3)}<\mu_2^{(1)}\le\mu_2^{(2)}\le\Sigma_2$, proving existence, simplicity and the gap bound.
To bound the overlaps, spectral decomposition gives, for normalized $\phi\in\Q(h_1)$,
\begin{equation}
 h_1[\phi]\ge E_1^{(1)}|\langle\psi^{(1)},\phi\rangle|^2
 +\mu_2^{(1)}(1-|\langle\psi^{(1)},\phi\rangle|^2).
 \label{eq:eckart}
\end{equation}
For each existing ground state, $h_1[\psi^{(j)}]\le E_1^{(j)}\le E_1^{(3)}$ for $j=2,3$ gives $F_{12},F_{13}\ge\kappa$; retaining $h_1[\psi^{(2)}]$ gives \eqref{eq:refinedoverlap}.
When $\kappa>0$, applying \eqref{eq:eckart} to $H_2$ and $\phi=\psi^{(3)}$ yields
\begin{equation}
 F_{23}\ge\frac{\mu_2^{(2)}-E_1^{(3)}}{\mu_2^{(2)}-E_1^{(2)}}\ge\kappa,
 \label{eq:F23}
\end{equation}
because this fraction increases with each unknown argument $\mu_2^{(2)}$ and $E_1^{(2)}$ in the indicated range.
\end{proof}

The three inputs in \eqref{eq:alloverlaps} are $E_1^{(1)}$, $\mu_2^{(1)}$ and $E_1^{(3)}$; \eqref{eq:refinedoverlap} additionally uses $h_1[\psi^{(2)}]$.
Thus the energy-bracket width divided by the reference gap bounds the infidelity $1-F$.
For bounded perturbations with spectral width below a simple reference ground-state gap, Seelmann's estimate is quadratic at small width-to-gap ratios~\cite{Seelmann2019}; section~\ref{sec:yukawa} compares the bounds explicitly.
The form-order proof also covers unbounded perturbations, unequal domains and lower-bounded self-adjoint operators on an abstract Hilbert space.
If $\psi^{(1)}$ can be prepared, $F_{12}\ge\kappa>0$ means that ideal quantum phase estimation for $H_2$ returns its ground-state energy with probability at least $\kappa$.
This bounds the success probability; the measurement precision and simulation cost require separate analysis.

\begin{example}[Shrinking Dirichlet boxes]\label{ex:shrinkingboxes}
For zero potentials on $(0,L_i)$ with $(L_1,L_2,L_3)=(1,0.9,0.8)$, zero extension gives the form order and $\kappa=[4-(L_1/L_3)^2]/3=0.8125$.
The normalized ground states $\sqrt{2/L}\sin(\pi x/L)$ have
\[
 F(L,M)=\frac{4q\sin^2(\pi q)}{\pi^2(1-q^2)^2},\qquad q=M/L<1.
\]
Thus $(F_{12},F_{13},F_{23})\simeq(0.96485,0.86434,0.95672)$ all exceed $\kappa$.
The domain changes are not finite bounded additive perturbations on a common domain.
\end{example}

\subsection{Low-energy subspaces and excited-state overlaps}
The subspace estimate uses spectral sums and the gap above the selected levels.
Let $P_n^{(i)}$ project onto an ordered orthonormal set of the first $n$ discrete eigenstates.
Leakage is invariant under basis changes within fixed subspaces; if a cutoff splits a degenerate level, the bound holds for every choice of the truncated subspace.
Put
\begin{equation}
 S_n^{(i)}=\sum_{j=1}^n\mu_j^{(i)},\qquad
 \xi_n=\frac{S_n^{(3)}-S_n^{(1)}}{\mu_{n+1}^{(1)}-\mu_n^{(1)}},\qquad \xi_0=0,
 \label{eq:xi}
\end{equation}
whenever the denominator is positive.

\begin{theorem}\label{thm:subspace}
Suppose \eqref{eq:order} holds.
If the first $n$ eigenvalues of all three operators are discrete and $\mu_{n+1}^{(1)}>\mu_n^{(1)}$, then, writing $\vartheta_k$ for the principal angles between the two subspaces,
\begin{equation}
 n-\tr(P_n^{(2)}P_n^{(1)})=\sum_{k=1}^n\sin^2\vartheta_k\le\xi_n.
 \label{eq:subspace}
\end{equation}
For $n\ge2$, if also $\mu_{n-1}^{(1)}<\mu_n^{(1)}$, every normalized $n$th eigenstate of $H_2$ satisfies
\begin{equation}
 |\langle\psi_n^{(1)},\psi_n^{(2)}\rangle|^2\ge1-\xi_n-\xi_{n-1}.
 \label{eq:excitedoverlap}
\end{equation}
The statements hold within any common reducing symmetry sector.
\end{theorem}

The underlying energy-excess estimate is the variational subspace bound expressed by Vu and Lei's curvature lemma~\cite[lemma~4.2]{VuLei2013}.
Taking $(\mu_{n+1}^{(1)}-H_1)_+$ reduces it to a finite-rank estimate even at a continuum threshold; form order then replaces the unknown target expectation by comparison spectral sums.
Appendix~\ref{app:overlap} gives a direct proof and obtains \eqref{eq:excitedoverlap} by splitting the components below and above the reference state.

The bound \eqref{eq:excitedoverlap} rules out orthogonality between the two excited states when its right-hand side is positive.
Ranking each different energy only once need not preserve the min--max order (appendix~\ref{app:distinct}).

\section{Ground-state distribution order from potential differences}
\label{sec:shape}
The shape of the potential difference can bound a whole class of spatial observables at once.
\subsection{Single crossing, likelihood-ratio order and observables}
In this section $\psi_i=\psi^{(i)}>0$ denotes the ground state of system $i$, not its $i$th eigenstate.
For one-dimensional, even and radial problems, let $p_i$ denote the normalized density of $x$, $|x|$ and $r$, respectively.
For a radial profile $\phi_i$ in $d$ dimensions, $p_i(r)=|S^{d-1}|r^{d-1}\phi_i(r)^2$.
Thus $p_2/p_1=(\psi_2/\psi_1)^2$ in the relevant variable.
The notation $p_2\lr p_1$ means that this ratio is nonincreasing.
This likelihood-ratio order implies $\ev{f}_2\le\ev{f}_1$ for every nondecreasing $f$ for which the expectations exist.

\begin{theorem}[Single-crossing comparison]\label{thm:ratio}
Let $V_1$ be admissible on $I=(a,b)$, and $V_2=V_1+D$ admissible on $(a,c)$ with $a<c\le b$, imposing Dirichlet conditions and extending $\psi_2$ by zero on $[c,b)$.
Suppose both discrete ground states exist.
If $D-(E_1^{(2)}-E_1^{(1)})$ changes sign at most once, from negative to positive, then $\psi_2/\psi_1$ is nonincreasing.
In particular this holds when $D$ is nondecreasing, irrespective of its sign.
For even potentials on intervals $(-R_i,R_i)$ with $0<R_2\le R_1\le\infty$, the same conclusion holds in $|x|$.
For central potentials on concentric balls of radii $R_i$ with the same nesting, it holds in $r$ in every dimension $d\ge2$.
In the even and radial cases, the potentials are admissible, both discrete ground states exist, and the single-crossing condition is imposed on $(0,R_2)$.
The ratio is constant only for equal domains and an almost everywhere constant $D$.
\end{theorem}

The proof uses $w=\psi_1\psi_2'-\psi_1'\psi_2$:
\begin{equation}
 w'=\frac{2m}{\hbar^2}\psi_1\psi_2(D-\Delta E),\qquad
 (\psi_2/\psi_1)'=w/\psi_1^2,\qquad \Delta E=E_1^{(2)}-E_1^{(1)}.
 \label{eq:wronskian}
\end{equation}
The single crossing makes $w$ decrease and then increase; its endpoint limits force $w\le0$.
Appendix~\ref{app:ratio} treats the endpoints and the radial origin.
For even potentials, this argument already appears in A~Martin's proof reproduced in~\cite{AHS1981}; the ratio formulation is also given in~\cite{HerbstPitt1991}.
The extension here includes asymmetric intervals, a shortened right endpoint and radial reference potentials without an additional curvature condition.
For smooth potentials on the line, the asymmetric case is implicit in the diffusion argument of~\cite{HerbstPitt1991}, section~10.1.
The monotonicity condition on $D$ can be checked directly; applying the more general single-crossing criterion requires information about $\Delta E$.

Applying the theorem twice gives
\begin{equation}
 p_3\lr p_2\lr p_1,\qquad
 \ev{f}_3\le\ev{f}_2\le\ev{f}_1,
 \label{eq:between}
\end{equation}
when $V_2-V_1$ and $V_3-V_2$ are both nondecreasing in the relevant variable.
Both inequalities reverse if both differences are nonincreasing, on a common domain.
For radial problems, this bounds every positive moment of the intermediate ground state by the two comparison states.
The direction of the differences matters: for
$V_1=x^2/2$, $V_2=V_1+a \mathrm{e}^{-x^2/s^2}$ and $V_3=V_1+a$ with $a,s>0$, the outside ground states coincide, but the middle state has larger strictly increasing moments of $|x|$.
Indeed $V_1-V_2$ increases in $|x|$, so theorem~\ref{thm:ratio} orders the reversed pair.

In the tilted box of figure~\ref{fig:distribution-nodes}(a), increasing the tilt does not decrease the cumulative probability to the left of any point; it need not order the probability densities pointwise.

\subsection{Coordinatewise comparison without reflection symmetry}
For $x,y\in\R^d$, let $(x\vee y)_j=\max(x_j,y_j)$ and $(x\wedge y)_j=\min(x_j,y_j)$.
A potential is submodular if
\begin{equation}
 V(x\vee y)+V(x\wedge y)\le V(x)+V(y)\qquad\text{for all }x,y.
 \label{eq:submodular}
\end{equation}
For smooth $V$, this is equivalent to $\partial_j\partial_kV\le0$ for $j\ne k$.
A positive function $u$ is multivariate totally positive of order two (MTP$_2$) if the reversed multiplicative inequality holds:
$u(x\vee y)u(x\wedge y)\ge u(x)u(y)$~\cite{KarlinRinott1980}.

\begin{theorem}\label{thm:lattice}
Let $\Omega=\R^d$ or an open box with Dirichlet conditions.
Suppose $V_1,D$ are continuous and bounded below, and $H_1=T+V_1$, $H_2=H_1+D$ have ground states below their essential spectra.
If $V_1$ is submodular and $D$ is coordinatewise nondecreasing, then $\psi_1$ is MTP$_2$ and $\psi_2/\psi_1$ is coordinatewise nonincreasing.
Their densities satisfy
\begin{equation}
 p_2(x)p_1(y)\le p_2(x\wedge y)p_1(x\vee y),\qquad
 \ev{A}_2\le\ev{A}_1
 \label{eq:lattice}
\end{equation}
for every bounded coordinatewise nondecreasing observable $A$.
\end{theorem}

Herbst and Pitt's full-space diffusion theory supplies stochastic monotonicity, while their Schr\"odinger applications use smooth reflection-symmetric potentials~\cite[theorems~1.1, 1.9--1.10 and~5.3]{HerbstPitt1991}.
The proof in appendix~\ref{app:lattice} applies positive-correlation inequalities to discretized paths and projects onto the ground state, treating continuous potentials and Dirichlet boxes directly.
The present assumption is submodularity on the full domain; a general radial potential need not satisfy it.

Without submodularity, a coordinatewise increasing perturbation can shift another coordinate in the opposite direction.
Take $V_1=(x^2+y^2)/2+cxy$, $0<c<1$, and $D=t(1+\tanh x)/2$, $t>0$.
Stationarity gives $\ev{\partial_xV_2}_2=\ev{\partial_yV_2}_2=0$, hence
$\ev{y}_2=tc\ev{(1+\tanh x)'/2}_2/(1-c^2)>0=\ev{y}_1$.
Truncating $y$ at $\pm M$ gives a bounded increasing observable that violates \eqref{eq:lattice} for large $M$.

\section{Node ordering and the limits of excited-state comparison}
\label{sec:nodes}
The $n$th one-dimensional bound state has exactly $n-1$ internal nodes~\cite{Simon2005}: zeros inside the interval, excluding Dirichlet endpoints.
We compare their positions using Sturm comparison.

\begin{theorem}\label{thm:nodes}
Let $V_1$ and $V_2=V_1+D$ be admissible on the same interval $I=(a,b)$, with $D$ nondecreasing, and suppose $E_n^{(i)}<\Sigma_i$.
Writing the nodes in increasing order gives
\begin{equation}
 z_k^{(2)}\le z_k^{(1)},\qquad k=1,\ldots,n-1.
 \label{eq:nodes}
\end{equation}
If also $D$ is bounded and the $n$th state stays below the essential spectrum along $H(t)=H_1+tD$, $0\le t\le1$, then
\begin{equation}
 \frac{\mathrm{d}z_k}{\mathrm{d}t}=
 \frac{2m}{\hbar^2\psi_t'(z_k)^2}
 \int_a^{z_k}\big(D-\ev{D}_{\psi_t}\big)\psi_t^2\,\mathrm{d}x\le0,
 \label{eq:velocity}
\end{equation}
where $\psi_t$ is the normalized real $n$th eigenfunction.
\end{theorem}

\begin{figure}[t]
\centering
\includegraphics[width=\linewidth]{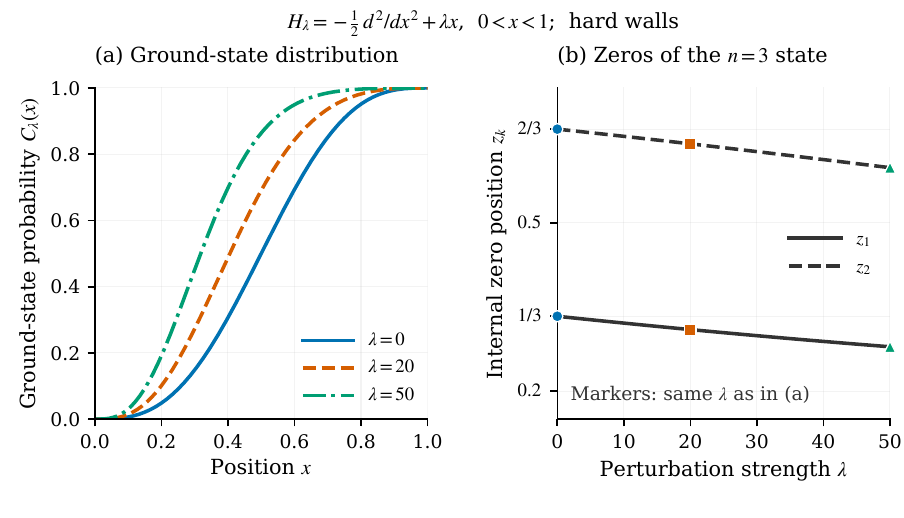}
\caption{Ground-state distributions and excited-state nodes of
$H_\lambda=-\tfrac12 \mathrm{d}^2/\mathrm{d}x^2+\lambda x$ on $(0,1)$, with hard walls and $\hbar=m=1$.
(a) The normalized ground-state cumulative probability
$C_\lambda(x)=\int_0^x|\psi_{1,\lambda}(s)|^2\,\mathrm{d}s$ is nondecreasing in $\lambda$, expressing a shift towards the left wall.
(b) The two internal zeros of the third eigenstate move left as $\lambda$ increases; boundary zeros are excluded. Colored markers match panel (a).
The distribution and node curves are numerical results checked against the Airy-function solution (appendix~\ref{app:numerics}).
The parameter $\lambda$ is a static perturbation strength.}
\label{fig:distribution-nodes}
\end{figure}

Guo and Zhang~\cite{GuoZhang2022} proved differentiability of nodes for finite-interval Sturm--Liouville problems with integrable potentials and obtained an equivalent derivative formula.
Integrating that formula and taking Dirichlet exhaustions yields \eqref{eq:nodes} under the present endpoint assumptions (appendix~\ref{app:nodes}).
We give a direct two-sided Sturm proof and use \eqref{eq:velocity} to describe node sensitivity along an isolated branch.

Write $\psi_i''+q_i\psi_i=0$, with $q_i=2m(E_n^{(i)}-V_i)/\hbar^2$.
Since $q_2-q_1$ is nonincreasing, Sturm comparison orders nodes from the left before its possible sign change and from the right after it.
The equal node counts complete the argument (appendix~\ref{app:nodes}).
Thus a uniform field in a finite well, $D=eFx$ with $eF>0$, shifts every internal node to the left, as illustrated for a tilted box in figure~\ref{fig:distribution-nodes}(b).
The parameter $t$ in \eqref{eq:velocity} is a perturbation strength, not physical time.
On a common interval symmetric about zero, for even $V_1$ and $D$ nondecreasing in $|x|$, apply the argument on the positive half-interval. Equal parity gives the same Dirichlet or Neumann condition at zero, so every corresponding positive node moves towards the origin.

Node ordering does not ensure a nonzero squared overlap for excited states.
For the unit box, $V_1=0$ and $V_2=h\mathbf1_{(c,1)}$, with $h\ge0$ and $0<c<1/2$, choose the second eigenfunctions to have positive derivative at zero.
Their overlap is $1$ at $h=0$ and tends, as $h\to\infty$, to
\begin{equation}
 -\frac{\sqrt c\sin(2\pi c)}{\pi(1-c^2)}<0.
 \label{eq:boxoverlap}
\end{equation}
For $j=1,2$, min--max bounds $E_j(h)$ by $j^2\pi^2/(2c^2)$, so the normalized states are bounded in $H_0^1(0,1)$ and their mass on $(c,1)$ tends to zero.
Compactness gives orthonormal limits supported on $(0,c)$; their eigenvalue equations and the min--max bounds identify them as its first two Dirichlet states.
The positive-derivative phase at zero fixes the sign, giving convergence of the second state.
The limiting overlap follows by integrating the normalized second box state on $(0,c)$ against $\sqrt2\sin(2\pi x)$.
Analytic perturbation for finite $h$~\cite{Kato1976} gives continuity, hence exact orthogonality at a finite $h$ despite a monotone potential difference.

\section{Yukawa--Hulth\'en comparisons for different state properties}
\label{sec:yukawa}
We choose Hulth\'en comparisons for energy and overlap, spatial observables, and the binding range in turn.
In atomic units, set
\begin{equation}
 V_{\rm Y}(r)=-\frac{\mathrm{e}^{-\delta r}}r,\qquad
 V_{\rm H}^{(\mu)}(r)=-\frac{\mu}{\mathrm{e}^{\mu r}-1},\qquad V_{\rm C}(r)=-\frac1r.
 \label{eq:potentials}
\end{equation}
The parameters $\delta,\mu>0$ are inverse screening lengths.
The Hulth\'en $s$ levels and its unnormalized ground-state radial function are
\begin{equation}
 E_n^{\rm H}(\mu)=-\frac{(1-n^2\mu/2)^2}{2n^2}\quad(n^2\mu<2),\qquad
 u_\mu(r)=\mathrm{e}^{-(1-\mu/2)r}(1-\mathrm{e}^{-\mu r}).
 \label{eq:hulthen}
\end{equation}
We use $u(r)=\sqrt{4\pi}\,r\psi(r)$ normalized in $L^2(0,\infty)$, so $p(r)=u(r)^2$ is the density of the radius, whereas $|\psi(0)|^2$ is the spatial density at the nucleus.

\subsection{Energy brackets and ground-state overlap}
For energy comparison, use the elementary inequalities
\begin{equation}
 V_{\rm H}^{(2\delta)}-c\delta\le V_{\rm Y}\le V_{\rm H}^{(2\delta)},\qquad
 c=\max_{x>0}\mathrm{e}^{-x}\left(\frac1x-\frac1{\sinh x}\right)=0.055810\ldots.
 \label{eq:hall}
\end{equation}
The right inequality is $\sinh x\ge x$ with $x=\delta r$; the left follows from the definition of $c$.
Together they bound each $s$-wave min--max value between $\mu_n^{\rm H}(2\delta)-c\delta$ and $\mu_n^{\rm H}(2\delta)$, where $\mu_n^{\rm H}=0$ when the corresponding bound level is absent.
The outside potentials share their eigenfunctions, giving a common reference state for the overlap bound.

For $0<\delta<1$, applying theorem~\ref{thm:overlap} to \eqref{eq:hall} in the $s$-wave sector gives
\begin{equation}
 F_{\rm Y,H}:=|\langle\psi_{\rm Y},\psi_{{\rm H},2\delta}\rangle|^2
 \ge1-\frac{c\delta}{g_{\rm H}},\qquad
 g_{\rm H}=\mu_2^{\rm H}(2\delta)-E_1^{\rm H}(2\delta).
 \label{eq:yukawafidelity}
\end{equation}
Here $\mu_2^{\rm H}(2\delta)=0$ for $\delta\ge1/4$.
The positive guarantee ends at $\delta_0=1+c-\sqrt{c(2+c)}\simeq0.7171$.
This bounded potential difference also permits Seelmann's stronger estimate~\cite[theorem~1.1]{Seelmann2019}.
Set $w=c\delta$, $A=H_{\rm H}^{(2\delta)}-wI$ and $B=H_{\rm Y}-A$.
Then $0\le B\le wI$ and $\|B\|=w$; the reference ground state is separated from the remaining spectrum by $g_{\rm H}$ and satisfies the convex-hull separation condition.
For $t=w/g_{\rm H}<1$, the perturbed rank-one spectral projection is the Yukawa ground-state projection, giving
\begin{equation}
 F_{\rm Y,H}\ge\frac{1+\sqrt{1-t^2}}2,\qquad t=\frac{c\delta}{g_{\rm H}}<1.
 \label{eq:yukawaseelmann}
\end{equation}
For $0<t<1$, this is strictly stronger than $1-t$ using the same $c$ and $g_{\rm H}$.
At $\delta=0.5$, the bounds are approximately $0.987382$ and $0.776760$, respectively, while the numerical overlap is $0.999759$.
Both positive-guarantee conditions require $\delta<\delta_0$; the reference ground state itself exists up to $\delta<1$.

\subsection{Radial moments and density at the nucleus}
To bound spatial observables, we choose Hulth\'en potentials whose differences from Yukawa are monotone.

\begin{corollary}\label{cor:yukawa}
For $0<\delta<2/\sqrt6$, all states in the chain
\begin{equation}
 V_{\rm C}\le V_{\rm H}^{(\delta)}\le V_{\rm Y}\le V_{\rm H}^{(\sqrt6\delta)}
 \label{eq:yukawaorder}
\end{equation}
have ground states, and each consecutive potential difference is nonincreasing in $r$.
Consequently
\begin{equation}
 p_{\rm C}\lr p_{{\rm H},\delta}\lr p_{\rm Y}\lr p_{{\rm H},\sqrt6\delta}.
 \label{eq:yukawadensities}
\end{equation}
For every $k>0$,
\begin{equation}
 \ev{r^k}_{{\rm H},\delta}\le\ev{r^k}_{\rm Y}\le\ev{r^k}_{{\rm H},\sqrt6\delta},\qquad
 \frac{4-6\delta^2}{4\pi}\le|\psi_{\rm Y}(0)|^2\le\frac{4-\delta^2}{4\pi}.
 \label{eq:momentcontact}
\end{equation}
The factor $\sqrt6$ is the smallest upper screening factor giving a nonincreasing potential difference.
\end{corollary}

Appendix~\ref{app:yukawa} proves the required monotonicity.
For $0<\mu<2$, the Hulth\'en moments are explicit:
\begin{equation}
 \ev{r^k}_{{\rm H},\mu}=
 \Gamma(k+1)\frac{a_0^{k+1}-2a_1^{k+1}+a_2^{k+1}}{a_0-2a_1+a_2},
 \quad a_j=(2-\mu+j\mu)^{-1},\quad k>0.
 \label{eq:momentformula}
\end{equation}
In particular, $\ev r_{{\rm H},\mu}=(2-\mu)^{-1}+1/2+(2+\mu)^{-1}$ and
$|\psi_{{\rm H},\mu}(0)|^2=(4-\mu^2)/(4\pi)$.
The moment inequalities reverse for $-3<k<0$; the moments there are defined by the convergent radial integral, avoiding the removable singularities of the separated Gamma terms at $k=-1,-2$.
For stronger screening, $p_{{\rm H},\delta}\lr p_{\rm Y}$ remains valid whenever the Yukawa ground state exists, since the lower Hulth\'en state must then exist.
This retains the positive-moment lower bounds and the contact-density upper bound; the lower contact bound is then simply zero.

In figure~\ref{fig:yukawa-bounds}, the upper mean-radius bound grows as $\delta\uparrow2/\sqrt6$ because the upper reference state approaches zero binding and develops a long tail.
The cutoff $\delta_0$ marks loss of the overlap guarantees, whereas $2/\sqrt6$ marks the end of the two-sided spatial comparison; neither is the Yukawa binding threshold.

\begin{figure}[t]
\centering
\includegraphics[width=\linewidth]{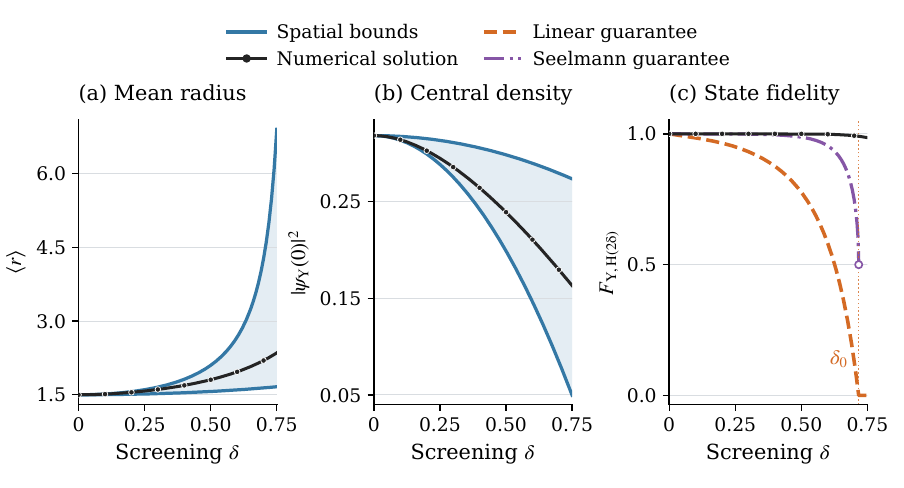}
\caption{Yukawa ground states in atomic units.
(a) Mean radius and (b) spatial density at the nucleus, with bounds \eqref{eq:momentcontact} from Hulth\'en screenings $\delta$ and $\sqrt6\delta$.
(c) Squared overlap with the Hulth\'en state at screening $2\delta$: numerical values (black), the linear lower bound \eqref{eq:yukawafidelity} truncated at zero (orange), and Seelmann's bound \eqref{eq:yukawaseelmann} (purple).
The dotted line marks $\delta_0\simeq0.7171$; the open endpoint excludes $t=1$ from Seelmann's condition.
The range $0\le\delta\le0.75$ lies below $2/\sqrt6$, with $\delta=0$ the Coulomb limit.
Black curves are wavefunction calculations; colored bounds are analytic.}
\label{fig:yukawa-bounds}
\end{figure}

\subsection{Optimizing the guaranteed binding range}
Define $\delta_c(n,0)$ as the supremum of screening parameters for which the $n$th Yukawa $s$ state is bound.
The upper potential in \eqref{eq:hall} gives $n^2\delta_c(n,0)\ge1$.
Varying both Hulth\'en strength and screening enlarges this guaranteed interval.

\begin{proposition}\label{prop:critical}
Let $m_*\in(1,2)$ be the root of $(2-m)\mathrm{e}^m=2$, and let $b_*=[\mathrm{e}(2-m_*)]^{-1}$.
Then $V_{\rm Y}\le b_*V_{\rm H}^{(m_*\delta)}$ and
\begin{equation}
 n^2\delta_c(n,0)\ge\frac{2b_*}{m_*}=1.13611\ldots.
 \label{eq:critical}
\end{equation}
This is the largest existence bound obtained from any pointwise majorant $bV_{\rm H}^{(m\delta)}$ with $b,m>0$, attained only at $(m,b)=(m_*,b_*)$.
\end{proposition}

The proof in appendix~\ref{app:critical} establishes the optimum analytically.
The improvement over $1$ applies to every $s$-level index $n\ge1$.
Brau and Calogero's counting bound~\cite{BrauCalogero2003}, with $g^2=2/\delta$, gives stronger existence guarantees at sufficiently large $n$, approaching the coefficient $4/\pi$.
That bound includes finite-$n$ corrections; \eqref{eq:critical} gives one explicit coefficient for all $n$, optimal within the stated comparison family.

\section{Conclusion}
The state property sought determines how to choose comparison potentials.
Comparison spectra and gaps control overlaps, whereas spatial order requires structure in the potential differences; ordered excited-state nodes can coexist with orthogonality.
For Yukawa ground states, this distinction leads to separate Hulth\'en choices for overlaps, radial observables and binding.
The bounded overlap comparison admits Seelmann's stronger guarantee, while optimizing strength and screening gives the largest binding range within the scaled Hulth\'en family.

\section*{Acknowledgments}
This work is supported by the National Natural Science Foundation of China (Grants 12475020 and 92565111), Quantum Science and Technology--National Science and Technology Major Project (2021ZD0301701), and the National Key Research and Development Program of China (2023YFC2205802).
\section*{Conflict of interest}
The authors declare no competing interests.
\section*{Data availability}
The numerical verification code is available on GitHub~\cite{VerificationCode}. The figure-reproduction code and sampled data are available from the corresponding authors upon reasonable request. No experimental data were generated.

\appendix
\section{Spectral-subspace estimate}\label{app:overlap}
Let $Q$ project onto the first $n$ eigenstates of $H_1$, and let $P$ be any rank-$n$ projection spanned by orthonormal form-domain vectors $\phi_j$.
Writing $M_{jj}=\langle\psi_j^{(1)},P\psi_j^{(1)}\rangle\in[0,1]$ and $\tau=\tr(PQ)$, spectral decomposition gives
\begin{equation}
 \sum_{j=1}^n h_1[\phi_j]\ge\sum_{j=1}^n\mu_j^{(1)}M_{jj}
 +\mu_{n+1}^{(1)}(n-\tau).
\end{equation}
Since $\sum_j\mu_j^{(1)}(1-M_{jj})\le\mu_n^{(1)}(n-\tau)$, taking $P=P_n^{(2)}$ and using $\sum_jh_1[\psi_j^{(2)}]\le S_n^{(3)}$ proves \eqref{eq:subspace}.
For a chosen $n$th eigenvector $\psi$ of $H_2$, choose its ordered eigenbasis so that $\psi\in\operatorname{Ran}P_n^{(2)}$ and $\psi\perp\operatorname{Ran}P_{n-1}^{(2)}$, also when the target level is degenerate.
Then
$\|(1-P_n^{(1)})\psi\|^2\le\xi_n$ and
$\|P_{n-1}^{(1)}\psi\|^2\le\xi_{n-1}$, by the two trace estimates.
Subtracting these squared norms from one gives \eqref{eq:excitedoverlap}.

\section{Endpoint argument for the ground-state ratio}\label{app:ratio}
By \eqref{eq:wronskian}, $w$ decreases and then increases, so its endpoint limits exist.
At an infinite endpoint, $w\in L^1$ by Cauchy--Schwarz and $\psi_i\in H^1$; its limit is therefore zero.
At a common finite Dirichlet endpoint $e$, Hardy's inequality gives $\int |w(x)|/|x-e|\,\mathrm{d}x<\infty$ locally near $e$, again forcing the limit to vanish.
At a shortened right endpoint $c<b$, a positive limit would make the positive ratio $\psi_2/\psi_1$ increase while tending to zero, which is impossible.
Thus both endpoint limits are nonpositive, and $w\le0$ throughout.
The constant-ratio case follows from $w'=0$ and positivity.
On the symmetric domains specified in the theorem, the positive ground states are even, and $w(0)=0$ replaces the left-end condition.

For radial profiles use
$w=r^{d-1}(\phi_1\phi_2'-\phi_1'\phi_2)$.
Its derivative is \eqref{eq:wronskian} multiplied by $r^{d-1}$.
At zero, Hardy's inequality gives $\liminf|w|=0$ for $d\ge3$.
For $d=2$, set $\eta(r_0)^2=\max_i\int_0^{r_0}|\phi_i'|^2r\,\mathrm{d}r\to0$.
Cauchy--Schwarz gives
$|\phi_i(r)|\le|\phi_i(r_0)|+\eta(r_0)\sqrt{\log(r_0/r)}$.
If $|w|\ge\gamma>0$ near zero, integration and the same inequality give
\[
\gamma\log(r_0/\varepsilon)\le2\eta(r_0)^2\log(r_0/\varepsilon)[1+o(1)]
\]
as $\varepsilon\to0$, a contradiction for small $r_0$.
The existing one-sided limit is consequently zero.
At the outer endpoint the previous argument applies with the radial weight.
The radial centrifugal terms also cancel, so the same proof works for the lowest state in a fixed angular-momentum sector.

\section{Coordinatewise order from imaginary time}\label{app:lattice}
Adding constants allows $V_1,D\ge0$.
Let $g(x)=\mathrm{e}^{-|x|^2}$, $\tau=\beta/N$, and $p_\tau$ be the kernel of $\mathrm{e}^{-\tau T}$.
For $z=(x_{-N},\ldots,x_N)$ define
\begin{equation}
 \Phi_i(z)=g(x_{-N})g(x_N)
 \prod_{k=-N}^{N-1}p_\tau(x_k,x_{k+1})
 \prod_{k=-N}^{N}\mathrm{e}^{-\tau c_kV_i(x_k)},
 \label{eq:pathdensity}
\end{equation}
Here $c_{\pm N}=1/2$ and $c_k=1$ for $|k|<N$.
On $\R^d$ the Gaussian heat kernel is MTP$_2$ in its two arguments.
On a box it is a product of totally positive one-dimensional Dirichlet kernels~\cite{KarlinMcGregor1959}.
Submodularity makes each potential factor in $\Phi_1$ MTP$_2$.
Products and marginals preserve this property~\cite{KarlinRinott1980}.
Writing $K_i=\mathrm{e}^{-\tau V_i/2}\mathrm{e}^{-\tau T}\mathrm{e}^{-\tau V_i/2}$ and $u_i=K_i^Ng$, the marginal at $x_0=x$ is $m_i(x)=u_i(x)^2$.

For $x\le x'$, the conditional densities of the remaining coordinates under $\Phi_1$ satisfy
\[
 f_{x'}(y\vee y')f_x(y\wedge y')\ge f_{x'}(y)f_x(y'),
\]
by MTP$_2$, since the conditional normalization factors cancel.
Expectations of decreasing functions are therefore smaller under $f_{x'}$~\cite{KarlinRinott1980}.
The factor $G=\Phi_2/\Phi_1=\exp(-\tau\sum c_kD(x_k))$ is decreasing in every coordinate.
It follows that
\begin{equation}
 \begin{aligned}
 \frac{m_2(x')}{m_1(x')}&=\int G(x',y)f_{x'}(y)\,\mathrm{d}y
 \le\int G(x,y)f_{x'}(y)\,\mathrm{d}y\\
 &\le\int G(x,y)f_x(y)\,\mathrm{d}y=\frac{m_2(x)}{m_1(x)}.
 \end{aligned}
\end{equation}
Hence $u_2/u_1$ is decreasing and $u_1$ is MTP$_2$.
Kato's product formula for form sums~\cite[theorem~18.2]{Simon2019Kato} gives
$Q_N=(B_NA_N^2)^N\to \mathrm{e}^{-\beta H_i}$ strongly, where $A_N=\mathrm{e}^{-\tau V_i/2}$ and $B_N=\mathrm{e}^{-\tau T}$.
For the symmetric product, write $Y_N=(B_NA_N^2)^{N-1}B_NA_N$, so $K_i^N=A_NY_N$, $Q_N=Y_NA_N$.
The contraction bounds give
$\|(K_i^N-Q_N)f\|\le\|(A_N-I)Q_Nf\|+\|(A_N-A_N^2)f\|\to0$.
For the first term, put $q=\mathrm{e}^{-\beta H_i}f$ and split
$(A_N-I)Q_Nf=(A_N-I)(Q_Nf-q)+(A_N-I)q$.
Both terms vanish by strong convergence and $\|A_N-I\|\le2$.
Hence $u_i\to \mathrm{e}^{-\beta H_i}g$ in $L^2$.
Normalization and $\beta\to\infty$ then give the positive ground states, since their eigenvalues are isolated and $\langle g,\psi_i\rangle>0$.
Almost-everywhere convergent subsequences preserve the ratio inequality for almost every strictly comparable pair and MTP$_2$ for almost every pair: the meet--join map permutes coordinates on each order region.
Continuity and positivity extend both inequalities to all pairs, including equal coordinates~\cite{Simon1982}.
Finally, $p_1$ MTP$_2$ and $p_2/p_1$ decreasing imply
$p_2(x)p_1(y)\le p_2(x\wedge y)p_1(x\vee y)$.
Another application of the same comparison inequality proves the expectation statement.

\section{Node comparison and differentiation}\label{app:nodes}
Put $q_i=2m(E_n^{(i)}-V_i)/\hbar^2$.
For consecutive zeros $\alpha<\beta$ of $\psi_1$, assume $q_2\ge q_1$ between them.
If $\psi_2$ has no zero in $(\alpha,\beta]$, choose both functions positive there.
Then $w'=\psi_1\psi_2(q_1-q_2)\le0$, but
$w(\beta)=-\psi_1'(\beta)\psi_2(\beta)>0$ and $w(\alpha^+)\le0$, a contradiction.
The endpoint case uses appendix~\ref{app:ratio}.
Reversing the interval shows that when $q_1\ge q_2$ between consecutive zeros of $\psi_2$, $\psi_1$ has a zero in $[\alpha,\beta)$.

For these one-dimensional Dirichlet--Friedrichs realizations, Sturm oscillation gives simple eigenvalues and exactly $n-1$ internal nodes for the $n$th bound state~\cite{Simon2005}.

Choose $x_*$ so that $q_2\ge q_1$ to its left and $q_2\le q_1$ to its right.
If $z_k^{(1)}\le x_*$, the $k$ disjoint intervals $(z_{j-1}^{(1)},z_j^{(1)}]$ each contain a node of $\psi_2$, so $z_k^{(2)}\le z_k^{(1)}$.
If $z_k^{(2)}\ge x_*$, the $n-k$ intervals $[z_j^{(2)},z_{j+1}^{(2)})$ each contain a node of $\psi_1$, giving the same conclusion.
Here $z_0^{(1)}=a$ and $z_n^{(2)}=b$; otherwise $z_k^{(2)}<x_*<z_k^{(1)}$ directly.

The finite-interval derivative formula also gives an exhaustion proof of the order.
On each compact interior interval, $V_i\in L^1$ and $D$ is bounded, so integration in $t$ orders the nodes~\cite{GuoZhang2022}.
Form-core approximation and min--max give spectral convergence for $H_1$ and $H_2$; isolation, simplicity and the local equation give normalized eigenfunction convergence in $L^2$ and locally in $C^1$.
The $n-1$ simple nodes therefore converge in order, using only the bound states of $H_1$ and $H_2$.

For bounded $D$, an isolated simple eigenfunction is analytic in the common operator graph norm~\cite{Kato1976}, hence in $H^1_0$ and locally in $C^1$ by the form bound and local equation.
Thus $E_n'(t)=\ev D_{\psi_t}$ and, with a dot denoting $\partial_t$, $W=\psi_t\dot\psi_t'-\psi_t'\dot\psi_t$ satisfies
$W'=(2m/\hbar^2)(D-E_n')\psi_t^2\in L^1$.
Its endpoint limits exist and vanish: Cauchy--Schwarz gives $W\in L^1$ at infinity, and Hardy's inequality gives $W/|x-e|\in L^1$ near a finite Dirichlet endpoint $e$.
In particular, $W(a^+)=0$.
Differentiating the node equation gives
$z_k'=-\dot\psi_t(z_k)/\psi_t'(z_k)=W(z_k)/\psi_t'(z_k)^2$.
This proves \eqref{eq:velocity}.
Its integral is nonpositive: $D-\ev D$ is nondecreasing, and its integral against $\psi_t^2$ over the whole interval is zero.

\section{Monotone Yukawa--Hulth\'en differences}\label{app:yukawa}
Write $x=\delta r$.
The first two differences in \eqref{eq:yukawaorder}, divided by $\delta$, are
$x^{-1}-(\mathrm{e}^x-1)^{-1}$ and $(\mathrm{e}^x-1)^{-1}-\mathrm{e}^{-x}/x$.
They are nonnegative and nonincreasing: differentiation reduces the first claim to $\sinh(x/2)\ge x/2$, and the second to
$(1-\mathrm{e}^{-x})/x\le(1+x)^{-1/2}$.
The latter follows from $1-\mathrm{e}^{-x}\le\log(1+x)\le x/\sqrt{1+x}$.

For $m\ge2$, $V_{\rm Y}\le V_{\rm H}^{(m\delta)}$ follows from
$\mathrm{e}^{(m-1)x}-\mathrm{e}^{-x}-mx\ge0$, whose value and first derivative vanish at zero and whose second derivative is positive.
The difference is nonincreasing exactly when
\begin{equation}
 \left(\frac{\sinh(mx/2)}{mx/2}\right)^2\ge\frac{\mathrm{e}^x}{1+x}.
 \label{eq:sinhcondition}
\end{equation}
Expansion at zero requires $m^2/12\ge1/2$, so $m\ge\sqrt6$.
For sufficiency, use
\begin{equation}
 \coth z-\frac1z\ge\frac{z}{3+z},\qquad z>0.
 \label{eq:langevin}
\end{equation}
Multiplication by $z(3+z)(\mathrm{e}^{2z}-1)>0$ reduces this inequality to
$S(z)=(2z-3)\mathrm{e}^{2z}+2z^2+4z+3\ge0$.
It follows from $S(0)=S'(0)=S''(0)=0$ and $S'''(z)=16z\mathrm{e}^{2z}\ge0$.
At $m=\sqrt6$, let
$F(x)=2\log[\sinh(mx/2)/(mx/2)]+\log(1+x)-x$.
Then $F(0)=0$ and \eqref{eq:langevin} gives
\begin{equation}
 F'(x)\ge\frac{x}{1+x/\sqrt6}-\frac{x}{1+x}>0\quad(x>0).
\end{equation}
This proves \eqref{eq:sinhcondition} for all $x>0$.
Theorem~\ref{thm:ratio} gives \eqref{eq:yukawadensities}.
A decreasing ratio of normalized radial densities starts at a value at least one; its limit at zero is the ratio of $u'(0)^2=4\pi|\psi(0)|^2$.
Normalization of \eqref{eq:hulthen} and radial integration then give \eqref{eq:momentcontact} and \eqref{eq:momentformula}.

\section{Global optimization of the binding criterion}\label{app:critical}
For $bV_{\rm H}^{(\mu)}$, put $y=\mathrm{e}^{-\mu r}$ and $u=y^a(1-y)f(y)$ with $a=\sqrt{-2E}/\mu$.
The radial equation becomes
\begin{equation}
 y(1-y)f''+[(2a+1)-(2a+3)y]f'+[2b/\mu-(2a+1)]f=0.
\end{equation}
Regularity at $r=0$ forces the decaying hypergeometric solution to terminate at degree $n-1$, giving $2b/\mu=n(2a+n)$, with $a>0$ for square integrability.
The resulting Jacobi polynomial has $n-1$ internal zeros, identifying the $n$th state by Sturm oscillation.
Thus $E_n=-(b-n^2\mu/2)^2/(2n^2)$ exists exactly when $n^2\mu<2b$.
At $n=1$, this also gives \eqref{eq:hulthen} directly.

The majorant condition is $b\le\beta(m)$, where
\begin{equation}
 \beta(m)=\inf_{x>0}\frac{\mathrm{e}^{-x}(\mathrm{e}^{mx}-1)}{mx}.
\end{equation}
For fixed $m$, the logarithm of the expression being minimized has second derivative
$x^{-2}[1-(mx/2)^2/\sinh^2(mx/2)]>0$.
At $m=m_*$ its derivative vanishes at $x=1$, so $\beta(m_*)=b_*$.
For every $m>0$, testing the infimum at $x=m_*/m$ gives
\begin{equation}
 \frac{2\beta(m)}m\le
 \frac{2(\mathrm{e}^{m_*}-1)}{m_*m}\mathrm{e}^{-m_*/m}
 \le\frac{2(\mathrm{e}^{m_*}-1)}{\mathrm{e} m_*^2}=\frac{2b_*}{m_*},
\end{equation}
because $y \mathrm{e}^{-y}\le \mathrm{e}^{-1}$.
Equality holds at $m=m_*$, proving global optimality.
If a feasible pair attains the optimal value of $2b/m$, then $m=m_*$ and $b=\beta(m_*)=b_*$, so the maximizing pair is unique.
The $n$th $s$ level of $bV_{\rm H}^{(m\delta)}$ exists exactly when $n^2m\delta<2b$; min--max then proves \eqref{eq:critical}.

\section{Numerical checks}\label{app:numerics}
The Yukawa calculations compare 40- and 50-function Slater Rayleigh--Ritz bases with radial finite differences and Richardson extrapolation.
Figure~\ref{fig:distribution-nodes} is checked against the tilted box's Airy-function solution at $\lambda=20,50$.
Figure~\ref{fig:yukawa-bounds} uses normalized radial states at 152 screening values, with larger-cutoff and Slater-basis checks at $\delta=0.1,0.2,0.5,0.75$; its squared overlap is computed independently of the analytic bound.
Basis and mesh differences assess convergence, not rigorous continuum error.
The public code supplies reference values, comparison constants and exact oscillator checks.
\section{Multiplicity and the ranking of distinct energies}\label{app:distinct}
\begin{example}[Distinct energies]\label{ex:distinct}
In units $\hbar=m=1$, take $V_1=(x^2+y^2)/2$, $V_2=V_1+\delta x^2$ and $V_3=(1/2+\delta)(x^2+y^2)$.
For $0<\delta<8/9$, the third different energy of $H_2$ is $(1+3\sqrt{1+2\delta})/2<3$, whereas that of $H_1$ is $3$.
For the reverse discrepancy, take $V_1=(x^2+y^2)/2-\delta x^2$, $V_2=(x^2+y^2)/2$, and $V_3=(x^2+y^2)/2+\delta x^2$.
At $\delta=0.3$, the middle operator's third different energy exceeds that of the upper operator despite the potential order.
These discrepancies concern ranks after removing multiplicities; the min--max order remains valid.
\end{example}

\end{document}